\documentclass[12pt,onecolumn]{IEEEtran}
\usepackage{amsmath ,amssymb ,euscript ,yfonts,psfrag,latexsym,dsfont,graphicx,bbm,color,amstext,wasysym,subfig,hyperref,pdfsync,url,flushend}
\graphicspath{{/},{figures/}}

\begin{document}
\newtheorem{thm}{Theorem}
\newtheorem{cor}[thm]{Corollary}
\newtheorem{lemma}[thm]{Lemma}
\newtheorem{prop}[thm]{Proposition}
\newtheorem{problem}[thm]{Problem}
\newtheorem{remark}[thm]{Remark}
\newtheorem{defn}[thm]{Definition}
\newtheorem{ex}[thm]{Example}

\def\spacingset#1{\def\baselinestretch{#1}\small\normalsize}

\title{Stochastic bridges of linear systems}
\author{Yongxin Chen and Tryphon Georgiou\thanks{\hspace*{-10pt} Department of Electrical and Computer Engineering, University of Minnesota, Minneapolis, MN 55455; {\{chen2468,tryphon\}@umn.edu}  \hspace*{3pt}Research supported in part by NSF, AFOSR and the Hermes-Luh endowment.}}

\maketitle
\begin{abstract}
We study a generalization of the Brownian bridge as a stochastic process
that models the position and velocity of inertial particles
between the two end-points of a time interval.
The particles experience random acceleration and are assumed to have
known states at the boundary. 
Thus, the movement of the particles can be modeled as an Ornstein-Uhlenbeck process 
conditioned on position and velocity measurements at the two end-points.
It is shown that optimal stochastic control provides
a stochastic differential equation (SDE) that generates
such a bridge as a degenerate diffusion process.
Generalizations to higher order linear diffusions are considered.
\end{abstract}

\section{Introduction}
The theoretical foundations on how molecular dynamics affect large scale properties of ensembles were layed down more than a hundred years ago. A most prominent place among mathematical models has been occupied by the Brownian motion which provides a basis for studying diffusion and noise~\cite{nelson1967dynamical,gillespie1996mathematics,Klebaner,vanhandel07}.
The Brownian motion is captured by the mathematical model of a Wiener process, herein denoted by $w(t)$. It represents the random motion of particles suspended in a fluid where their inertia is negligible compared to viscous forces. Taking into account inertial effects under a ``delta-correlated'' stationary Gaussian force field $\eta(t)$ (that is, white noise, loosely thought of as $dw/dt$ \cite[p.\ 46]{nelson1967dynamical})
\[
m\frac{d^2x(t)}{dt^2}=-\lambda \frac{dx(t)}{dt} + \eta(t)
\]
represents the Langevin dynamics ; $x$ represents position, $m$ mass, $t$ time, and $\lambda$ viscous friction parameter. The corresponding SDE
\[
\left[\begin{matrix} dx(t)\\dv(t)\end{matrix}\right]=
\left[\begin{matrix} 0&1\\0 &-\lambda/m\end{matrix}\right]
\left[\begin{matrix} x(t)\\v(t)\end{matrix}\right]dt+
\left[\begin{matrix} 0\\1/m\end{matrix}\right]dw(t),
\]
where $w$ is a Wiener process and $v$ the velocity,
is a degenerate diffusion in that the stochastic term does not affect all degrees of freedom.

Sample paths of diffusion processes between end-point conditions is fundamental and have been considered since the early days of probability theory. A standard textbook example for a stochastic process ``pinned'' at the end-points of an interval, e.g., $x(0)=x(1)=0$, is the so-called Brownian bridge \cite[p.\ 35]{revuz1999continuous},  which has a well-known representation via the SDE (see \cite[p.\ 132]{Klebaner})
\[
dx(t)=-\frac{1}{1-t} x(t) dt + dw(t).
\]
Herein, motivated by transport of particles, we study bridges of general diffusion processes. In particular, we are interested in an SDE representation for an Ornstein-Uhlenbeck bridge where both position and velocity are pinned at the two ends of an interval. Such a ``pinned'' process is very natural when considering trasport of inertial particles in regimes where viscous forces are negligible (e.g., in rarefied gas dynamics).
We are also motivated by the relevance of such degenerate diffusion processes in interpolation of density functions (e.g., probability distributions of many particle systems,  power spectral distributions etc., cf.\ \cite{mccann1997convexity,villani2008optimal,jiang2012geometric})

Important connections between bridges of non-degenerate diffusion processes, large deviations in sample-path spaces, and optimal control have been studied \cite{daipra91,Pavon_nonequilibrium,PavonWakolbinger}. Interestingly, it appears that similar connections may be present for certain degenerate diffusion processes as well (cf.\ \cite{Pavon_nonequilibrium}). In fact, herein, we explain that for the Ornstein-Uhlenbeck bridge as well as for bridges of general linear time-varying dynamical systems, an SDE representation is always available. The SDE is constructed by solving the stochastic optimal control problem to ensure end-point conditions (see also, \cite{Kosmol_Pavon}).
To this end, we first explain the Brownian bridge in a way that will be echoed in the construction of an SDE for the Ornstein-Uhlenbeck bridge, followed by the construction of an SDE for bridges of general linear time-varying systems.

\section{Brownian bridge}
The standard Brownian bridge is typically defined as a stochastic process $\xi$ on $[0,\,1]$ with $\xi(0)=\xi(1)=0$, continuous sample paths, and values that are jointly normally distributed with $E\{\xi(t)\xi(s)\}=t(1-s)$ for $0\leq t\leq s\leq 1$. Alternatively, it is often defined as a stochastic process with the same statistics as $w(t)-tw(1)$ and continuous sample paths. Below we explain how to compute the statistics starting from the assumption that the process is pinned at $1$.

\subsection{Statistics of the Brownian bridge}
The Brownian bridge can be viewed as a
standard Wiener process $w$
on $[0,\,1]$ conditioned on $w(1)=0$. For $t\leq s$, as before, we have that the covariance of values of the Wiener process is
\[
E\{\left[\begin{matrix}w(t)\\ w(s)\\ w(1)\end{matrix}\right]\left[\begin{matrix}w(t),\; w(s),\; w(1)\end{matrix}\right]\}
=\left[\begin{matrix} t & t & t\\ t & s & s\\ t & s &1\end{matrix}\right].
\]
Therefore, the distribution of $\left[w(t),\; w(s)\right]^\prime$ conditioned on $w(1)=0$ is normal with zero mean and covariance
\[
\left[\begin{matrix} t(1-t) & t(1-s) \\ t(1-s) & s(1-s)\end{matrix}\right].
\]
This covariance and joint normality of the values provide the law for the Brownian bridge which agrees with those of the aforementioned definitions.

\subsection{Optimal control and SDE representation}
Now consider the linear-quadratic optimal control problem to minimize
\begin{equation}\label{eq:optimization}
J(t)=\int_t^1 u(\tau)^2d\tau,
\end{equation}
subject to $d\xi(t)/dt =u(t)$ and $\xi(1)=0$.
For the more familiar form of a cost functional with a terminal cost,
\[
J_F(t)=F\xi(1)^2+\int_t^1 u(\tau)^2d\tau
\]
with $d\xi(t)/dt =u(t)$
and $F>0$, the minimal values is
$p(t)\xi(t)^2$
with optimizing choice for the control being 
\[
u_{\rm opt}(t)=-p(t)\xi(t)
\]
and $p(t)$ satisfying the Riccati equation
$\dot p(t)=p^2(t)$ with boundary condition $p(1)=F$. Hence, we obtain the minimal value $(1-t)\xi^2(t)$ of \eqref{eq:optimization} as the limiting case when $F\to \infty$,
with the optimal choice for the control input
\begin{equation}\label{eq:optimalcontrol}
u_{\rm opt}(t) = -\frac{1}{1-t}\xi(t).
\end{equation}
The corresponding ``controlled'' SDE
\begin{eqnarray}\nonumber
d\xi &=& u_{\rm opt}(t)dt +dw(t)\\
&=& -\frac{1}{1-t}\xi(t)dt +dw(t),\label{eq:BB}
\end{eqnarray}
with $\xi(0)=0$, generates a Brownian bridge as can be easily verified \cite[p.\ 132]{Klebaner}. Indeed, the state transition of the deterministic time-varying system
\[
\frac{d\xi}{dt} = -\frac{1}{1-t}\xi(t) +r(t),
\]
which for this first order system coincides with the response at $s$ to an impulse at $t$, is
\[
\Phi(s,t)= \frac{1-s}{1-t}.
\]
It follows that the solution to \eqref{eq:BB} has a representation as a stochastic integral
\[
\xi(t)=\int_0^t \frac{1-t}{1-\tau}dw(\tau).
\]
and therefore, assuming $t\leq s$,
\begin{align*}
E\{\xi(t)\xi(s)\}&=\int_0^t\frac{(1-t)(1-s)}{(1-\tau)^2}d\tau\\
&=t(1-s).
\end{align*}
This proves that indeed, \eqref{eq:BB} is a Brownian bridge.

\section{Ornstein-Uhlenbeck bridge}\label{sec:LangevinB}

We now follow exactly the same steps in order to define a bridge for the Ornstein-Uhlenbeck dynamics. Without loss of generality we assume that there are no viscous forces and the mass normalized to one.
Thus, we begin with the SDE
\begin{subequations}\label{eq:LangevinB}
\begin{align}
d\xi(t)&=
\left[\begin{array}{cc} 0 &1\\ 0 & 0\end{array}\right]
\xi(t)dt + \left[\begin{array}{c} 0\\ 1\end{array}\right]dw(t)
\end{align}
where
\[
\xi(t)=\left[\begin{array}{c} x(t)\\ v(t)\end{array}\right]
\]
is the vectorial process composed of the position and velocity components. We now condition these to satisfy
an initial and a final condition,
\begin{align}
\xi(0)&=0 \mbox{ and }
\xi(1)=0,
\end{align}
respectively.
\end{subequations}

\subsection{Statistics of the Ornstein-Uhlenbeck bridge}

To determine the statistics dictated by \eqref{eq:LangevinB} we condition the ``velocity'' $v(t)$, which in this case is a Wiener process, since $dv(t)=dw(t)$, to satisfy
\begin{subequations}
\begin{eqnarray}
v(0)&=&0\\
v(1)&=&0\\
x(1)=\int_0^1v(\tau)&=&0,
\end{eqnarray}
\end{subequations}
while it is given that $x(0)=0$.
To this end, we first consider the covariance of the vector
\[
\left[\begin{matrix} v(t)& v(s)& v(1)& x(1)\end{matrix}\right]^\prime,
\]
readily seen to be
\[
\left[\begin{matrix} t & t & t & t-\frac{t^2}{2}\\ t & s & s & s-\frac{s^2}{2}\\ t & s &1 & \frac{1}{2}\\ t-\frac{t^2}{2} & s-\frac{s^2}{2} & \frac{1}{2} & \frac{1}{3} \end{matrix}\right].
\]
Therefore, the covariance of $\left[\begin{matrix} v(t)& v(s)\end{matrix}\right]^\prime$ when conditioned on $\left[\begin{matrix} v(1)& x(1)\end{matrix}\right]^\prime$ being the zero vector, can be evaluated as the Schur complement
\begin{eqnarray*}
&&
\left[\begin{matrix} t & t \\\\ t & s\end{matrix}\right]
-
\left[\begin{matrix} t & t-\frac{t^2}{2}\\ \\s & s-\frac{s^2}{2} \end{matrix}\right]
\left[\begin{matrix} 1 & \frac{1}{2}\\\\ \frac{1}{2} & \frac{1}{3} \end{matrix}\right]^{-1}
\left[\begin{matrix}t & s \\\\ t-\frac{t^2}{2} & s-\frac{s^2}{2}\end{matrix}\right].
\end{eqnarray*}
This is
\begin{eqnarray*}
\left[\begin{array}{cc}
-t(3t^3 - 6t^2 + 4t - 1) &
-t(s - 1)(3st - 3s + 1)\\
-t(s - 1)(3st - 3s + 1) & -s(3s^3 - 6s^2 + 4s - 1)
\end{array}\right].
\end{eqnarray*}

\subsection{Optimal control and SDE representation}

Just like in the case of the Brownian bridge, we now consider the linear-quadratic optimal control problem to minimize
    \[
        \xi(1)^\prime F\xi(1)+\int_0^1u(\tau)'u(\tau)d\tau
    \]
subject to
\begin{align*}
d\xi(t)&=
\left[\begin{array}{cc} 0 &1\\ 0 & 0\end{array}\right]
\xi(t)dt + \left[\begin{array}{c} 0\\ 1\end{array}\right]u(t)dt.
\end{align*}
By solving the corresponding Riccati equation and taking the limit as $F\to\infty$, we obtain the optimal control
\[
u(t)=-\left[\begin{matrix}\frac{6}{(1-t)^2} &\frac{4}{1-t}\end{matrix}\right]\xi(t)
 \]
for the problem to minimize
$\int_0^1u(\tau)'u(\tau)d\tau$ subject to a terminal condition $\xi(1)=0$.
This will be further explained in Section~\ref{sec:LinearB} for the more general case of linear time-varying dynamics.

We now consider the corresponding ``controlled'' SDE
\begin{align}\label{eq:LangevinBridge}
d\xi(t)=&
\left[\begin{array}{cc} 0 &1\\ -\frac{6}{(1-t)^2} & -\frac{4}{1-t}\end{array}\right]
\xi(t)dt + \left[\begin{array}{c} 0\\ 1\end{array}\right]dw(t).
\end{align}
We claim that \eqref{eq:LangevinBridge} realizes the Ornstein-Uhlenbeck bridge. To establish this, we need to show that the statistics of solutions to \eqref{eq:LangevinBridge} are consistent with those of the ``pinned'' process generated by \eqref{eq:LangevinB} derived earlier.
That is, for $\xi(t)^\prime=[x(t),\;v(t)]$ 
it suffices to show that for solutions of \eqref{eq:LangevinBridge},
    \[
        E\{v(t)v(t)\}=-t(3t^3 - 6t^2 + 4t - 1)
    \]
and
    \[
        E\{v(t)v(s)\}=-t(s - 1)(3st - 3s + 1).
    \]
Since $x(t)$ is $\int_0^tv(\tau)d\tau$ in both cases, the statistics of $x(t)$ will also be consistent.
The proof is given in Section \ref{sec:LinearB} for the more general case of time-varying linear dynamics.

\section{The bridge for a time-varying linear system}\label{sec:LinearB}

We consider the linear SDE
    \begin{subequations}\label{eq:LinearB}
    \begin{align}\label{eq:LinearB1}
        d\xi(t)& =A(t)\xi(t)dt+B(t)dw(t)
    \end{align}
with initial condition
    \begin{align}\label{eq:initial}
        \xi(0) &=0,
    \end{align}
and are interested in solutions that are conditioned to satisfy
    \begin{align}\label{eq:final}
        \xi(1) &=0
    \end{align}
    \end{subequations}
as well. Below, we first determine the statistics of the pinned process and then an SDE that generates the bridge.

\subsection{Statistics of the bridge}\label{sec:stats}

Since~\eqref{eq:LinearB1} is a linear SDE driven by Wiener process and $\xi(0)=0$, it follows that $\xi(t)$ is a zero-mean Gaussian process. Thus, we only need to determine second order statistics of the conditioned process.
The covariance of
    \[
        \left[\begin{matrix} \xi(t)'& \xi(s)'& \xi(1)'\end{matrix}\right]^\prime
    \]
is
    \begin{equation}\label{eq:Covariance}
        \left[\begin{matrix}
        P(t)& P(t)\Phi(s,t)'& P(t)\Phi(1,t)'\\
        \Phi(s,t)P(t) & P(s) & P(s)\Phi(1,s)'\\
        \Phi(1,t)P(t) & \Phi(1,s)P(s) & P(1)
        \end{matrix}\right]^\prime,
    \end{equation}
where $\Phi(s,t)$ is the state transition of~\eqref{eq:LinearB1} and
    \[
        P(t)=E\{\xi(t)\xi(t)'\}
    \]
satisfies the Lyapunov equation
    \begin{equation}\label{eq:Cov}
        \dot{P}(t)=A(t)P(t)+P(t)A(t)'+B(t)B(t)'.
    \end{equation}
Since $\xi(0)=0$ is given, $P(0)=0$.
Taking the Schur complement of~\eqref{eq:Covariance} gives the covariance of $\left[\begin{matrix} \xi(t)'& \xi(s)'\end{matrix}\right]^\prime$ conditioned on $\xi(1)=0$ as
    \[
        \left[\begin{matrix}
        Q(t,t) & Q(t,s)\\
        Q(t,s)' & Q(s,s)
        \end{matrix}\right],
    \]
where
    \begin{align}\label{eq:CrossCov}
        Q(t,s)&=P(t)\Phi(s,t)'-P(t)\Phi(1,t)'P(1)^{-1}\Phi(1,s)P(s).
    \end{align}
Any stochastic process that agrees with these statistics will be referred to as a bridge of \eqref{eq:LinearB}.

\subsection{SDE representation}

Once again let us consider the linear-quadratic optimization problem to minimize
    \[
        \xi(1)'F\xi(1)+\int_0^1u(\tau)'u(\tau)d\tau
    \]
subject to the dynamics
    \[
        d\xi(t) =A(t)\xi(t)dt+B(t)u(t)dt.
    \]
The optimal solution is $u_{\rm opt}(t)=-B(t)'\hat{P}(t)^{-1}\xi(t)$ where $\hat{P}(t)$ satisfies the differential Lyapunov equation
    \begin{equation}\label{eq:Covdual}
        \dot{\hat{P}}(t)=A(t)\hat{P}(t)+\hat{P}(t)A(t)'-B(t)B(t)'
    \end{equation}
with boundary condition $\hat{P}(1)=F^{-1}$.
We consider the limiting case of infinite terminal cost, i.e., $F\rightarrow \infty$, corresponding to $\hat{P}(1)=0$ and verify that the corresponding controlled stochastic system realizes the sought bridge.
\begin{prop}\label{prop:prop1}
 Under the earlier notation and assumptions on $A,B,\hat{P},w$, the SDE
    \begin{equation}\label{eq:Closedloop}
        d\xi(t) =(A(t)-B(t)B(t)'\hat{P}(t)^{-1})\xi(t)dt+B(t)dw(t)
    \end{equation}
generates a bridge of \eqref{eq:LinearB}.
\end{prop}

\begin{proof}
We only need to consider second order statistics of solutions to \eqref{eq:Closedloop} and establish that
these coincide with the statistics computed in Section \ref{sec:stats}. Hence, for $0\le t\le s\le 1$ we denote $\hat{Q}(t,s)=E\{\xi(t)\xi(s)'\}$ to be the covariance of solutions to \eqref{eq:Closedloop} and we will show that $\hat Q(t,s)=Q(t,s)$.
For simplicity we denote $\hat Q(t,t)=\hat Q(t)$ and the same for $Q$.

We first begin with
    \begin{equation}\label{eq:AutoCov}
        Q(t)=P(t)-P(t)\Phi(1,t)'P(1)^{-1}\Phi(1,t)P(t)
    \end{equation}
and show that it also satisfies the differential Lyapunov equation
    \begin{equation}\label{eq:AutoCovDiff}
        \dot{Q}(t)=\hat{A}(t)Q(t)+Q(t)\hat{A}(t)'+B(t)B(t)'
    \end{equation}
    for
    \[
    \hat A(t) = (A(t)-B(t)B(t)'\hat{P}(t)^{-1}),
    \]
 and,  since
 $Q(0)=0$, that indeed $Q(t)=\hat Q(t)$.
To this end, consider $Q(t)$ as in \eqref{eq:AutoCov}. Then,
    \begin{align*}
        &\dot{Q}(t)-\hat{A}(t)Q(t)-Q(t)\hat{A}(t)'-B(t)B(t)'\\
        =\;\;& B(t)B(t)'G(t)+G(t)'B(t)B(t)',
    \end{align*}
where
    \begin{eqnarray*}
        G(t)&=&\hat{P}(t)^{-1}Q(t)-\Phi(1,t)'P(1)^{-1}\Phi(1,t)P(t)\\
        &=&\hat{P}(t)^{-1}P(t)-\hat{P}(t)^{-1}T(t)\Phi(1,t)'P(1)^{-1}\Phi(1,t)P(t)
            \end{eqnarray*}
    and
    \[
    T(t)=P(t)+\hat{P}(t).
    \]
    From \eqref{eq:Cov} and \eqref{eq:Covdual},
    \[
        \dot{T}(t)=A(t)T(t)+T(t)A(t)',
    \]
   and therefore
   \[
        T(t)=\Phi(t,0)T(0)\Phi(t,0)',
    \]
while $T(0)=\hat{P}(0)$ and $T(1)=P(1)$.
Since
    \begin{align*}
        &T(t)\Phi(1,t)'P(1)^{-1}\Phi(1,t)\\
        =&\Phi(t,0)T(0)\Phi(t,0)'\Phi(1,t)'P(1)^{-1}\Phi(1,t)\\
        =&\Phi(t,1)\Phi(1,0)T(0)\Phi(1,0)'P(1)^{-1}\Phi(1,t)\\
        =&\Phi(t,1)T(1)P(1)^{-1}\Phi(1,t)=I,
    \end{align*}
the identity matrix, we deduce that
    \[
        G(t)=\hat{P}(t)^{-1}P(t)-\hat{P}(t)^{-1}IP(t)=0.
    \]
Therefore \eqref{eq:AutoCovDiff} holds and $Q(t)=\hat Q(t)$.

For general $0\leq t\leq s\leq 1$,
    \[
        \hat{Q}(t,s)=\hat{Q}(t,t)\hat{\Phi}(s,t)'
    \]
where
    \[
        \frac{\partial \hat{\Phi}(s,t)}{\partial s}=\hat{A}(s)\hat{\Phi}(s,t).
    \]
Therefore,
    \[
        \frac{\partial \hat{Q}(t,s)}{\partial s}=\hat{Q}(t,s)\hat{A}(s)'.
    \]
We now show that $Q(t,s)$ satisfies the same differential equation, i.e., that
    \begin{equation}\label{eq:Partial}
        \frac{\partial Q(t,s)}{\partial s}=Q(t,s)\hat{A}(s)'.
    \end{equation}
From \eqref{eq:CrossCov} we deduce that
    \begin{eqnarray*}
        &&\frac{\partial Q(t,s)}{\partial s}-Q(t,s)\hat{A}(s)'
        =H(t,s)B(s)B(s)'
    \end{eqnarray*}
where
    \begin{eqnarray*}
        H(t,s)&=&Q(t,s)\hat{P}(s)^{-1}-P(t)\Phi(1,t)'P(1)^{-1}\Phi(1,s)\\
        &=&P(t)\Phi(s,t)'\hat{P}(s)^{-1}-P(t)K(t,s)\hat{P}(s)^{-1}.
    \end{eqnarray*}
But
    \begin{eqnarray*}
    K(t,s)&=&\Phi(1,t)'P(1)^{-1}\Phi(1,s)T(s)\\
          &=&\Phi(1,t)'P(1)^{-1}\Phi(1,s)\Phi(s,0)T(0)\Phi(s,0)'\\
          &=&\Phi(1,t)'P(1)^{-1}T(1)\Phi(s,1)'=\Phi(s,t)'.
    \end{eqnarray*}
Therefore $H(t,s)=0$ and \eqref{eq:Partial} holds. Since we already know that $Q(t,t)=\hat{Q}(t,t)$, it follows that $Q(t,s)=\hat{Q}(t,s)$. This completes the proof.
\end{proof}

\section{Bridge with arbitrary boundary points}

So far we have discussed bridges with initial and terminal states being $0$. The more general case with nonzero initial and terminal states is straightforward. More specifically, we consider the linear SDE
    \begin{subequations}\label{eq:GLinearB}
    \begin{align}\label{eq:GLinearB1}
        d\xi(t)& =A(t)\xi(t)dt+B(t)dw(t)
    \end{align}
with initial condition
    \begin{align}\label{eq:Ginitial}
        \xi(0) &=\xi_0,
    \end{align}
whle the process $\xi(t)$ is conditioned to satisfy
    \begin{align}\label{eq:Gfinal}
        \xi(1) &=\xi_1.
    \end{align}
    \end{subequations}
Below, we determine the statistics of the pinned process and then the SDE that generates the bridge.

\subsection{Statistics of the bridge}

The second order statistics of~\eqref{eq:GLinearB} coincide with those of~\eqref{eq:LinearB}. Hence, we only need to compute first-order statistics.
Considering only \eqref{eq:GLinearB1} and \eqref{eq:Ginitial},
    \[
        E\{\xi(t)\}=\Phi(t,0)\xi_0.
    \]
Thus, the conditional expectation of $\xi(t)$, given $\xi(1)=\xi_1$, is
    \begin{equation}\label{eq:Firstorder}
        L(t)=\Phi(t,0)\xi_0+P(t)\Phi(1,t)'P(1)^{-1}(\xi_1-\Phi(1,0)\xi_0).
    \end{equation}

\subsection{SDE representation}

In order to enforce the terminal constraint \eqref{eq:Gfinal},
we penalize the difference between $\xi(1)$ and $\xi_1$ and  consider the linear-quadratic optimal control problem to minimize
    \[
        J_F=(\xi(1)-\xi_1)'F(\xi(1)-\xi_1)+\int_0^1u(\tau)'u(\tau)d\tau
    \]
subject to the dynamics
    \[
        d\xi(t) =A(t)\xi(t)dt+B(t)u(t)dt.
    \]
The optimal solution is 
    \[
        u_{\rm opt}(t)=-B(t)'\hat{P}(t)^{-1}(\xi(t)-\Phi(t,1)\xi_1)
    \]
where $\hat{P}(t)$ satisfies the differential Lyapunov equation \eqref{eq:Covdual} with boundary condition $\hat{P}(1)=F^{-1}$. Once again the limit as $F\to\infty$ corresponds to $\hat{P}(1)=0$. We now verify that the resulting ``controlled'' SDE realizes the sought bridge.
\begin{prop}
Under the above assumptions on $A$, $B$, $\hat{P}$, and $w$, the SDE
    \begin{eqnarray}\nonumber
        d\xi(t) &=&\hat{A}(t)\xi(t)dt+B(t)B(t)'\hat{P}(t)^{-1}\Phi(t,1)\xi_1 dt\\
        && +B(t)dw(t)\label{eq:Closedloop1}
    \end{eqnarray}
with
\[
\hat{A}(t)=A(t)-B(t)B(t)'\hat{P}(t)^{-1}
\]
generates a bridge of \eqref{eq:GLinearB}.
\end{prop}
\begin{proof}
The second order statistics of~\eqref{eq:Closedloop1} coincide with those of~\eqref{eq:Closedloop} and, by Proposition \ref{prop:prop1} with those of \eqref{eq:LinearB}  and therefore \eqref{eq:GLinearB} as well. Next we show that the first order statistics are also consistent. 
For this, it suffices to show that $L(t)$ in~\eqref{eq:Firstorder} satisfies
    \[
        \dot{L}(t)=\hat{A}(t)L(t)+B(t)B(t)'\hat{P}(t)^{-1}\Phi(t,1)\xi_1.
    \]
Using the same argument as in the proof of Proposition~\ref{prop:prop1} we obtain
    \begin{eqnarray*}
        &&\dot{L}(t)-\hat{A}(t)L(t)-B(t)B(t)'\hat{P}(t)^{-1}\Phi(t,1)\xi_1\\
        &=&B(t)B(t)'\hat{P}(t)^{-1}(\Phi(t,1)(\xi_1-\Phi(1,0)\xi_0)
        \\&&+\,\Phi(t,0)\xi_0-\Phi(t,1)\xi_1)\\
        &=&0.
    \end{eqnarray*}
This completes the proof.
\end{proof}

\section{Illustrative examples}

We consider a double integrator as in Section~\ref{sec:LangevinB} with state $\xi(t)=[x(t)~ v(t)]'$, and plot two representative sample paths of
\eqref{eq:LangevinBridge}.
More specifically, Figure~\ref{fig:Position} and Figure~\ref{fig:Velocity} show position and velocity, respectively, while Figure~\ref{fig:Phase} shows the two paths in phase space. Phase plots of a 2-dimensional Brownian bridge are shown in Figure~\ref{fig:Brownian} for comparison.
\begin{figure}\begin{center}
    \includegraphics[width=0.47\textwidth]{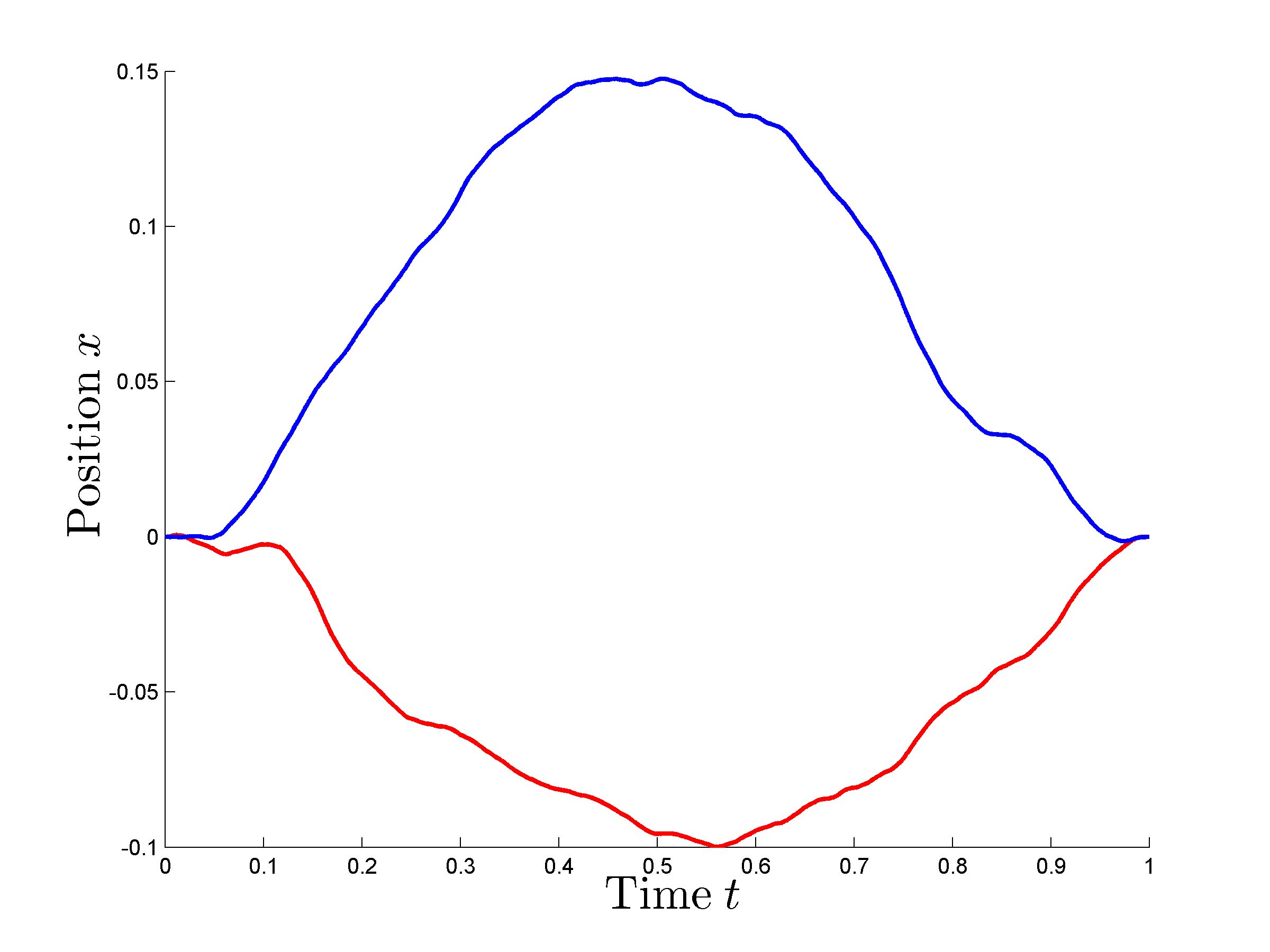}
    \caption{Position $x(t)$ of Ornstein-Uhlenbeck bridge sample paths}
    \label{fig:Position}
\end{center}\end{figure}
\begin{figure}\begin{center}
    \includegraphics[width=0.47\textwidth]{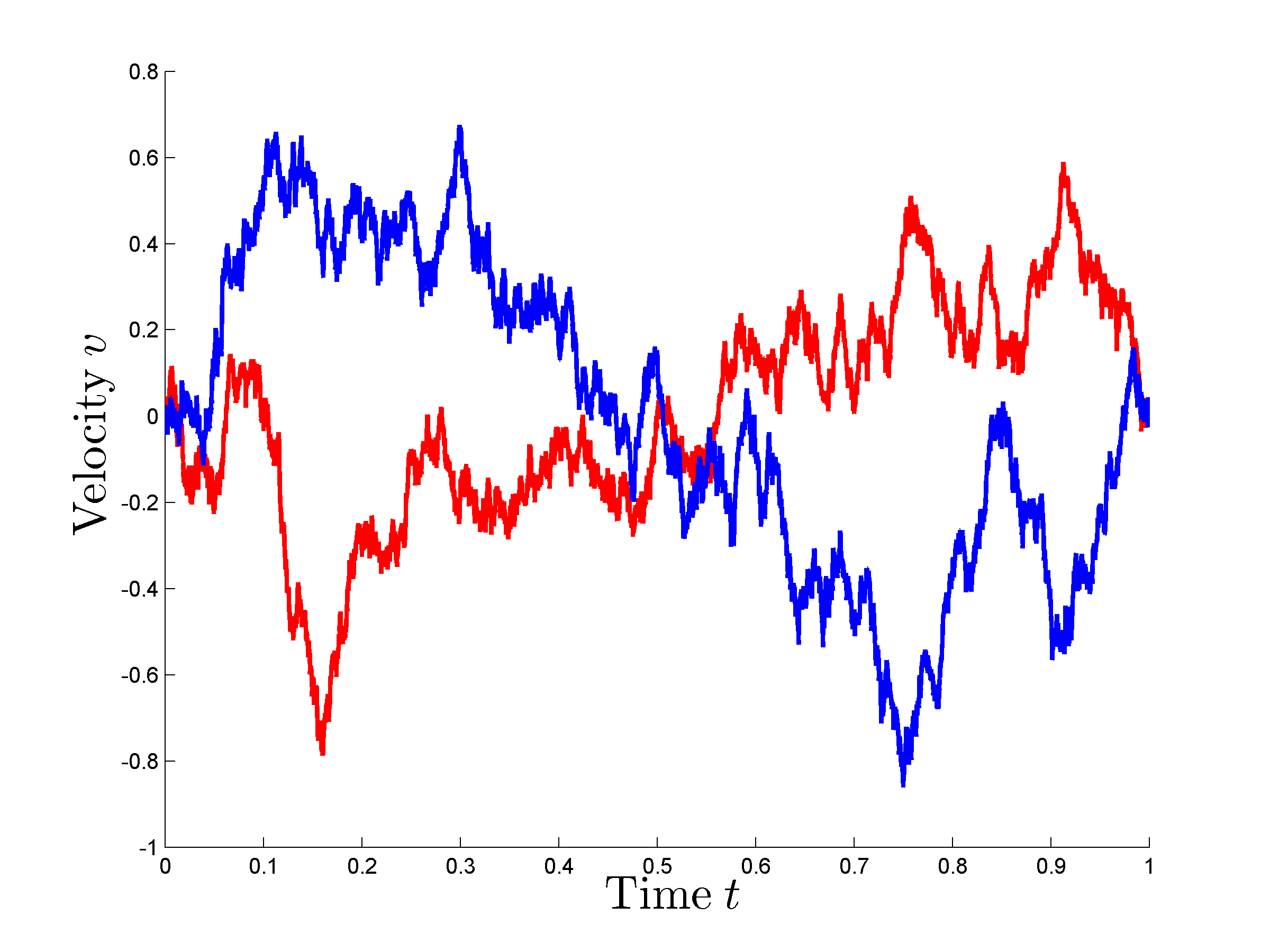}
    \caption{Velocity $v(t)$ of Ornstein-Uhlenbeck bridge sample paths}
    \label{fig:Velocity}
\end{center}\end{figure}
\begin{figure}\begin{center}
    \includegraphics[width=0.47\textwidth]{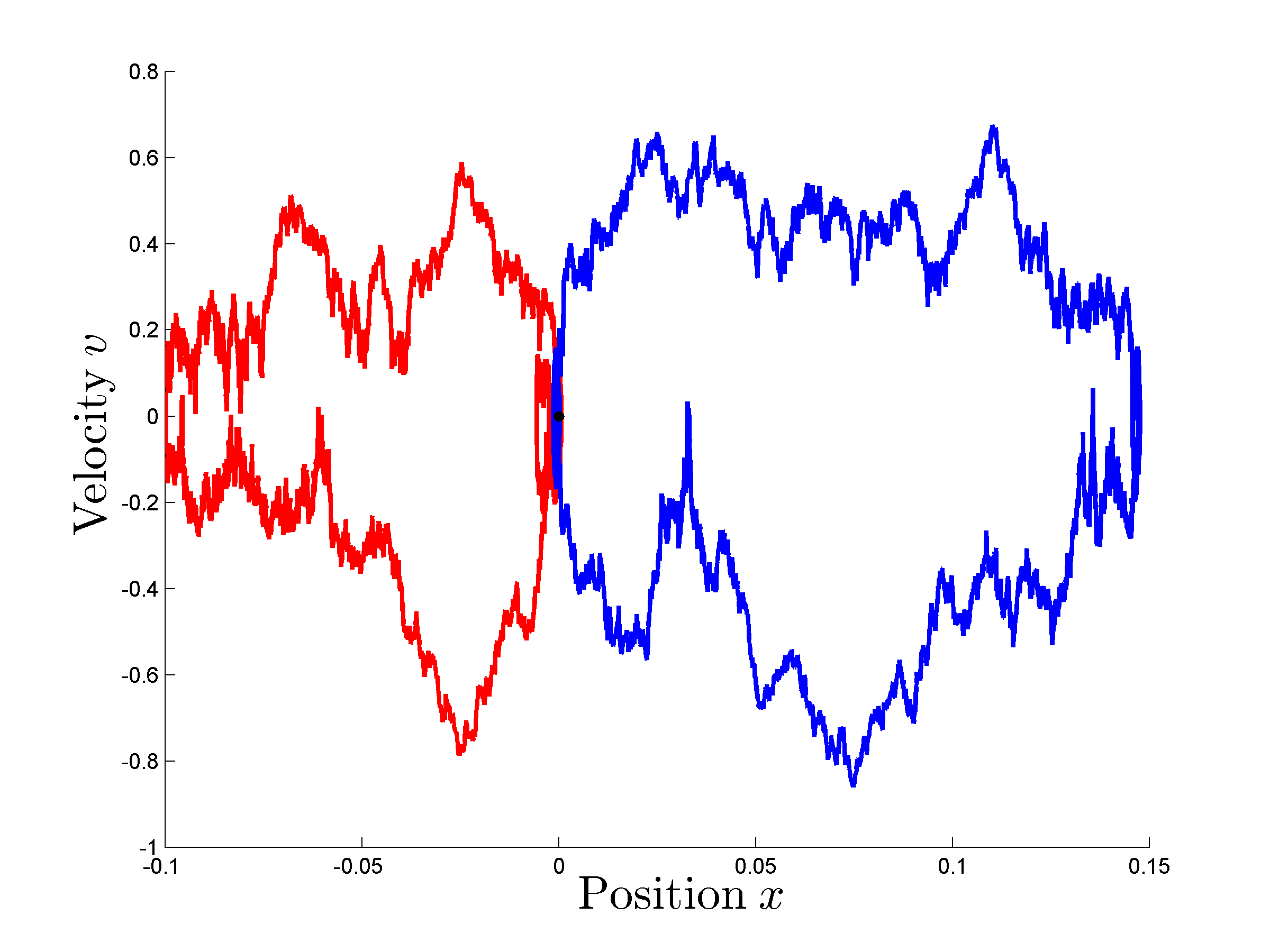}
    \caption{Phase plots of Ornstein-Uhlenbeck bridge sample paths}
    \label{fig:Phase}
\end{center}\end{figure}
\begin{figure}\begin{center}
    \includegraphics[width=0.47\textwidth]{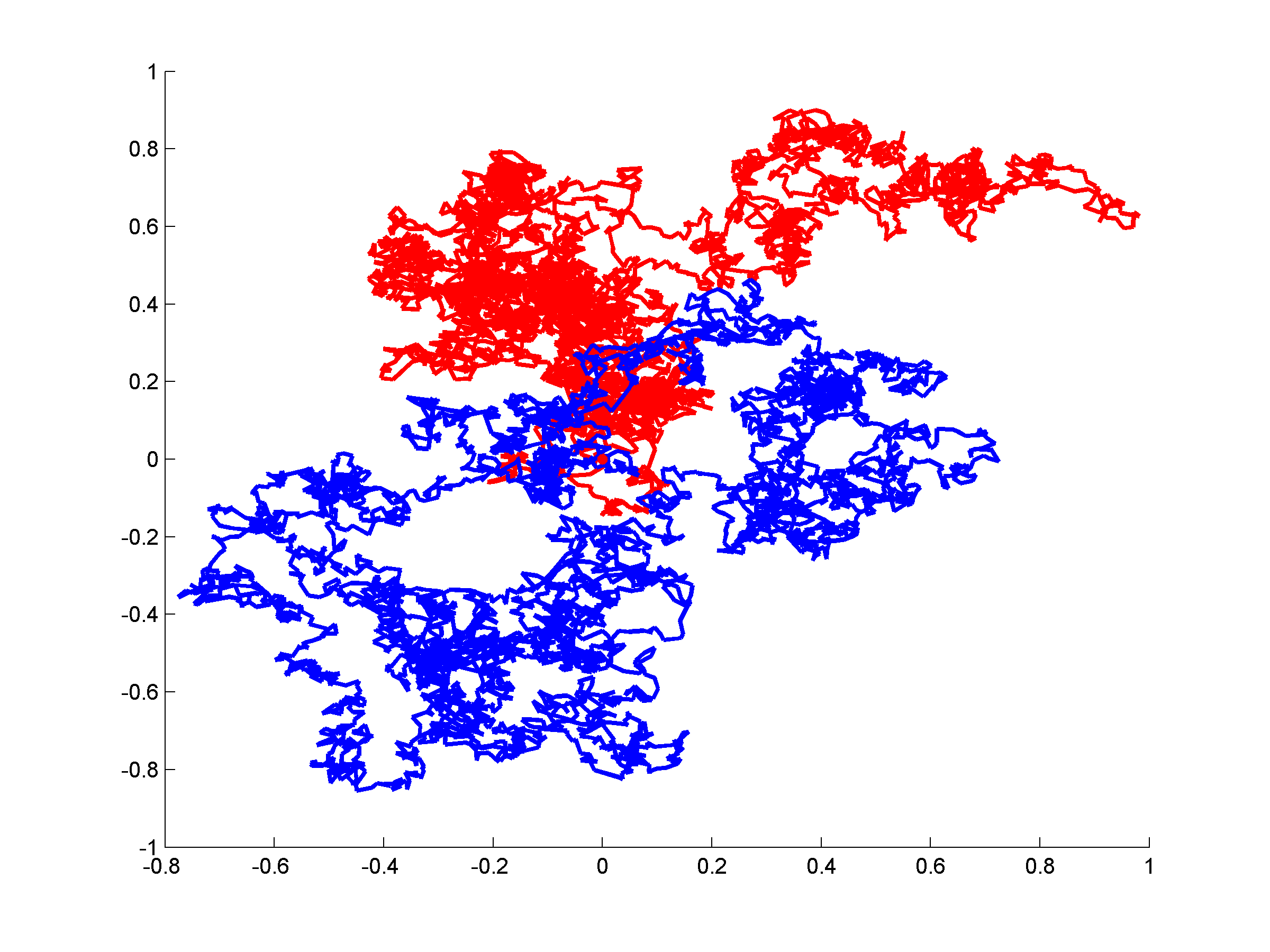}
    \caption{Phase plots of 2-dimensional Browian bridge sample paths}
    \label{fig:Brownian}
\end{center}\end{figure}

\section{Conclusion}

The Ornstein-Uhlenbeck bridge represents a ``pinned'' process with Ornstein-Uhlenbeck dynamics. We introduced such a process and a corresponding realization via a suitable SDE. The latter is constructed based on an optimal control problem. Generalization to bridges of linear diffusion processes is also presented. Our original aim has been to study possibly ways to interpolate density functions (probability distributions of many-particle systems,  power spectral distributions, and so on) and develop suitably geometric ideas
\cite{jiang2012geometric,jiang2012distances} in the spirit of \cite{mccann1997convexity,villani2008optimal}. The example of a pinned process is a first step towards a more general Sch\"odinger bridge as a possible such mechanism (see \cite{PavonWakolbinger} and the references therein) and this will be the subject of future work.

\section{Acknowledgment}
We would like to thank Michele Pavon for his input and for many inspiring discussions.

\bibliographystyle{IEEEtran}
\bibliography{refs}

\begin{thebibliography}{10}
\providecommand{\url}[1]{#1}
\csname url@samestyle\endcsname
\providecommand{\newblock}{\relax}
\providecommand{\bibinfo}[2]{#2}
\providecommand{\BIBentrySTDinterwordspacing}{\spaceskip=0pt\relax}
\providecommand{\BIBentryALTinterwordstretchfactor}{4}
\providecommand{\BIBentryALTinterwordspacing}{\spaceskip=\fontdimen2\font plus
\BIBentryALTinterwordstretchfactor\fontdimen3\font minus
  \fontdimen4\font\relax}
\providecommand{\BIBforeignlanguage}[2]{{%
\expandafter\ifx\csname l@#1\endcsname\relax
\typeout{** WARNING: IEEEtran.bst: No hyphenation pattern has been}%
\typeout{** loaded for the language `#1'. Using the pattern for}%
\typeout{** the default language instead.}%
\else
\language=\csname l@#1\endcsname
\fi
#2}}
\providecommand{\BIBdecl}{\relax}
\BIBdecl

\bibitem{nelson1967dynamical}
E.~Nelson, \emph{Dynamical theories of {B}rownian motion}.\hskip 1em plus 0.5em
  minus 0.4em\relax Princeton university press Princeton, 1967, vol.~17.

\bibitem{gillespie1996mathematics}
D.~T. Gillespie, ``The mathematics of {B}rownian motion and {J}ohnson noise,''
  \emph{American Journal of Physics}, vol.~64, no.~3, pp. 225--239, 1996.

\bibitem{Klebaner}
F.~C. Klebaner \emph{et~al.}, \emph{Introduction to stochastic calculus with
  applications}.\hskip 1em plus 0.5em minus 0.4em\relax World Scientific, 2005,
  vol.~57.

\bibitem{vanhandel07}
R.~van Handel, ``Stochastic calculus, filtering, and stochastic control,''
  \emph{Course notes: \url{http://www.princeton.edu/~rvan/acm217/ACM217.pdf}},
  2007.

\bibitem{revuz1999continuous}
D.~Revuz and M.~Yor, \emph{Continuous martingales and {B}rownian motion}.\hskip
  1em plus 0.5em minus 0.4em\relax Springer, 1999, vol. 293.

\bibitem{mccann1997convexity}
R.~J. McCann, ``A convexity principle for interacting gases,'' \emph{advances
  in mathematics}, vol. 128, no.~1, pp. 153--179, 1997.

\bibitem{villani2008optimal}
C.~Villani, \emph{Optimal transport: old and new}.\hskip 1em plus 0.5em minus
  0.4em\relax Springer, 2008, vol. 338.

\bibitem{jiang2012geometric}
X.~Jiang, Z.-Q. Luo, and T.~T. Georgiou, ``Geometric methods for spectral
  analysis,'' \emph{Signal Processing, IEEE Transactions on}, vol.~60, no.~3,
  pp. 1064--1074, 2012.

\bibitem{daipra91}
P.~Dai~Pra, ``A stochastic control approach to reciprocal diffusion
  processes,'' \emph{Applied mathematics and Optimization}, vol.~23, no.~1, pp.
  313--329, 1991.

\bibitem{Pavon_nonequilibrium}
M.~Pavon, ``Stochastic control and nonequilibrium thermodynamical systems,''
  \emph{Applied Mathematics and Optimization}, vol.~19, no.~1, pp. 187--202,
  1989.

\bibitem{PavonWakolbinger}
M.~Pavon and A.~Wakolbinger, ``On free energy, stochastic control, and
  {S}chr{\"o}dinger processes,'' in \emph{Modeling, Estimation and Control of
  Systems with Uncertainty}.\hskip 1em plus 0.5em minus 0.4em\relax Springer,
  1991, pp. 334--348.

\bibitem{Kosmol_Pavon}
P.~Kosmol and M.~Pavon, ``Lagrange approach to the optimal control of
  diffusions,'' \emph{Acta Applicandae Mathematica}, vol.~32, no.~2, pp.
  101--122, 1993.

\bibitem{jiang2012distances}
X.~Jiang, L.~Ning, and T.~T. Georgiou, ``Distances and {R}iemannian metrics for
  multivariate spectral densities,'' \emph{Automatic Control, IEEE Transactions
  on}, vol.~57, no.~7, pp. 1723--1735, 2012.

\end{thebibliography}
\end{document}